\documentclass[10pt]{article}

\usepackage{amsmath,amssymb,amsthm,fullpage,graphicx}

\DeclareMathOperator{\Exp}{{\Bbb E}}

\numberwithin{equation}{section}

\newtheorem{theorem}{Theorem}
\newtheorem{lemma}[theorem]{Lemma}
\newtheorem{conjecture}{Conjecture}

\newtheorem{corollary}[theorem]{Corollary}

\newtheorem{fact}{Fact}

\RequirePackage{dsfont}
\newcommand{\one}{\mathds{1}}
\renewcommand{\vec}[1]{\mathbf{#1}}

\newcommand{\remove}[1]{}

\newcommand{\C}{\mathbb{C}}

\newtheorem{prop}{Proposition}

\newcommand{\ket}[1]{\left| #1 \right\rangle}
\newcommand{\bra}[1]{\left\langle #1 \right|}
\renewcommand{\outer}[2]{\ket{#1} \!\bra{#2}}
\newcommand{\inner}[2]{\left\langle #1 \!\mid\! #2 \right\rangle}
\newcommand{\innerm}[3]{\left\langle #1 \!\mid\! #2 \!\mid\! #3 \right\rangle}

\DeclareMathOperator{\tr}{{\rm tr}}

\DeclareMathOperator{\rank}{{\rm rk}}

\newcommand{\abs}[1]{\left| #1 \right|}

\newcommand{\poly}{{\rm poly}}

\newcommand{\vsat}{V_{\rm sat}}
\newcommand{\rsat}{R_{\rm sat}}
\newcommand{\radv}{\rsat^{\rm adv}}
\newcommand{\rgen}{\rsat^{\rm gen}}
\newcommand{\aadv}{\alpha^{\rm adv}}
\newcommand{\dmax}{{d_{\max}}}

\newcommand{\e}{{\rm e}}

\newcommand{\dt}{{\rm d}t}
\newcommand{\db}{{\rm d}b}
\newcommand{\dx}{{\rm d}x}

\newcommand{\dmu}{{\rm d}\mu}
\newcommand{\dnu}{{\rm d}\nu}
\newcommand{\dtau}{{\rm d}\tau}
\newcommand{\QMA}{{\rm QMA}}
\newcommand{\NP}{{\rm NP}}
\newcommand{\numP}{{\rm \#P}}
\renewcommand{\P}{{\rm P}}

\newcommand{\be}{\begin{equation}}
\newcommand{\ee}{\end{equation}}
\newcommand{\ba}{\begin{array}}
\newcommand{\ea}{\end{array}}
\newcommand{\bea}{\begin{eqnarray}}
\newcommand{\eea}{\end{eqnarray}}

\newcommand{\calH}{{\cal H }}

\newcommand{\calS}{{\cal S }}

\begin{document}

\title{Bounds on the quantum satisfiability threshold}
\author{Sergey Bravyi\footnote{IBM Watson Research Center, Yorktown Heights NY 10594.}
\and
Cristopher Moore\footnote{University of New Mexico and the Santa Fe Institute.}
\and
Alexander Russell\footnote{University of Connecticut}}
\maketitle

%SBB: style changes
\abstract{Quantum $k$-SAT is the problem of deciding whether there is a $n$-qubit state
 which is perpendicular to a set of vectors, each of which lies in the Hilbert space of $k$ qubits.  Equivalently,
 the problem is to decide whether a particular type of local Hamiltonian has a ground state with zero energy.
 We consider random quantum $k$-SAT formulas with $n$ variables and $m=\alpha n$ clauses, and ask at what value of $\alpha$ these formulas cease to be satisfiable.  We show that the threshold for random quantum 3-SAT is
%at least $0.917$ and
at most $3.594$.  For comparison, convincing arguments from statistical physics suggest that the classical 3-SAT threshold is $\alpha_c \approx 4.267$.  For larger $k$, we show that the quantum threshold is a constant factor smaller than the classical one.
%In fact, our lower bound is the \emph{precise} threshold for the existence of a separable satisfying assignment when the clause vectors are separable.
Our bounds work by determining the generic rank of the satisfying subspace for certain gadgets, and then using the technique of differential equations to analyze various algorithms that partition the hypergraph into a collection of these gadgets.  
%This is the first time that differential equations have been used to establish upper bounds on a satisfiability threshold, and our techniques may apply to various classical problems as well.
%We also consider random quantum 3-SAT formulas where the forbidden vectors are separable.  Our upper bounds apply in this case as well.  We also obtain a lower bound by determining the exact threshold $\alpha_c^{\rm sep} \approx 0.917$ below which these formulas have a separable satisfying assignment.
}

%Laumann et al.\ proved a number of important structural features of quantum $k$-SAT.  In particular, they showed that for any fixed hypergraph $G$ of clauses, if the clause vectors are chosen randomly, then the rank of the satisfying subspace takes a particular value with probability $1$, which is a function only of $G$.  Any particular choice of clause vectors can only increase this rank, so if an adversary can make a choice of clause vectors such that the formula is unsatisfiable, then it is always unsatisfiable.  This includes classical basis vectors corresponding to $k$-SAT clauses.  As a consequence, if we consider random $k$-SAT formulas with

\section{Introduction}
\label{sec:intro}

In quantum $k$-SAT~\cite{bravyi}, each clause corresponds to a projection operator on the Hilbert space of $n$ qubits,
\[
C = (\one_k - \outer{v}{v} ) \otimes \one_{n-k} \, .
\]
Here $\ket{v}$ is a vector in the $2^k$-dimensional Hilbert space of some $k$-tuple of qubits, $\one_k$ is the identity on that Hilbert space, and $\one_{n-k}$ is the identity on the remaining qubits.  A formula is a set of clauses $\phi = \{C_1,\ldots,C_m\}$.  We say that $\phi$ is \emph{satisfiable} if there
%SBB: |x> usually denotes a basis vector. I changed x to \psi
 is a state $\ket{\psi}$ which is perpendicular to all the forbidden vectors $\ket{v}$: in other words, if
\[
\innerm{\psi}{C_i}{\psi} = 1 \mbox{ for all $i$}.
\]
%SBB: do we really need the equation below?
%or equivalently if
%\[
%\innerm{\psi}{\Phi}{\psi} = 1 \mbox{ where } \Phi = \prod_{i=1}^m C_i \, .
%\]
We call the subspace of such states $\ket{\psi}$ the \emph{satisfying subspace} $\vsat$.

In addition to being the quantum analogue of a canonical \NP-complete problem, quantum $k$-SAT is an illustrative case of a $k$-local Hamiltonian.
 %SBB:
  In that case, $\vsat$ is the subspace spanned by eigenstates of a Hamiltonian $H=\sum_i (I-C_i)$ with zero energy.
 %SBB: style change
 It was shown in~\cite{bravyi} that the decision problem of whether a particular quantum $k$-SAT formula is satisfiable is in \P\ for $k=2$, and is $\QMA_1$-complete for $k \ge 4$, where $\QMA_1$ is the subclass of $\QMA$ where the probability of acceptance for a yes-instance is $1$.

We are also interested in the problem of determining the
%SBB: 'rank of a subspace' is defined as its dimension
rank $\rsat = \dim \vsat$ of the satisfying subspace, or equivalently the degeneracy of the zero-energy ground states.  Determining $\rsat$ is a natural quantum analogue of a classical counting problem, namely determining the number of satisfying assignments of a $k$-SAT formula.  Classically, even for $k=2$ this problem is \numP-complete under Turing reductions~\cite{valiant}.  In the quantum case, it is not obvious that finding $\rsat$ is even in \numP, since the satisfying states may be arbitrarily entangled and may have no succinct description.  Indeed, it seems to us that one can define a natural quantum version of \numP\ as the class of problems consisting of finding the rank of an eigenspace of a $k$-local Hamiltonian, although we do not pursue this further here.

In the classical setting, a lively collaboration between computer scientists, statistical physicists, and mathematicians has grown up around the behavior of random $k$-SAT formulas.  These are constructed in the following way.  In order to construct a random formula $\phi(n,m)$ with $n$ variables and $m$ clauses, we first construct a random $k$-uniform hypergraph with $n$ vertices and $m$ edges, by choosing $m$ times, uniformly and with replacement, from the ${n \choose k}$ possible $k$-tuples of vertices.  Then, for each edge, we choose uniformly from the $2^k$ possible combinations of signs for those $k$ literals.

We are particularly interested in the sparse case, where $m=\alpha n$ for some constant $\alpha$.  There is a conjectured phase transition, where these formulas go from satisfiable to unsatisfiable when $\alpha$ exceeds a critical threshold:
\begin{conjecture}
For each $k \ge 3$, there is a constant $\alpha_c$ such that
\[
\lim_{n \to \infty} \Pr[ \mbox{$\phi(n,\alpha n)$ is satisfiable} ]
= \begin{cases}
1 & \mbox{if $\alpha < \alpha_c$} \\
0 & \mbox{if $\alpha > \alpha_c$} \, .
\end{cases}
\]
\end{conjecture}
In the absence of a proof of this conjecture, one can prove statements of the form that $\phi$ is unsatisfiable with high probability if $\alpha > \alpha^*$, or satisfiable with high probability if $\alpha < \alpha^\ddagger$.  Then, assuming that a phase transition exists, $\alpha^*$ and $\alpha^\ddagger$ are upper and lower bounds on the threshold $\alpha_c$.  The state of the art for classical 3-SAT is~\cite{3satlower1,3satlower2,3satupper}
\[
3.52 \le \alpha_c \le 4.490 \, ,
\]
although overwhelmingly convincing arguments from physics~\cite{mezard} indicate that
\[
\alpha_c \approx 4.267 \, .
\]

In the quantum case, we can similarly define a random quantum $k$-SAT formula $\phi(n,m)$ as a random hypergraph, where for each edge we choose the forbidden vector $\ket{v}$ uniformly from the vectors of norm $1$ in the Hilbert space $\C_2^{\otimes k}$ of those $k$ qubits.  We can then conjecture an analogous phase transition at a critical density $\alpha_c^q$.  Laumann et al.~\cite{laumann} showed that
\[
0.818... \le \alpha_c^q \le \alpha_c \, ,
\]
where $0.818...$ is the density at which the hypergraph has a nonempty 2-core with high probability.
% CM: added this
In this paper, we show that
\[
%0.917 \le
\alpha_c^q \le 3.594 \, .
\]
Note that this upper bound is well below the accepted value of the classical 3-SAT threshold.  We also show that for all $k \ge 4$,
\[
\alpha_c^q \le 2^k b \, ,
\]
where $b \approx 0.573$.  Since the classical threshold grows as $2^k \ln 2$~\cite{ach-peres}, this shows that the ratio $\alpha_c^q / \alpha_c$ is strictly less than $1$.

In order to prove these results, we exploit the observation of Laumann et al.\ that, once the hypergraph $G$ is fixed, $\rsat$ takes a generic value $\rgen$ with probability $1$.  One way to see this is to note that with probability $1$, the components of the clause vectors are algebraically independent transcendentals.  Then any subdeterminant of the matrix of forbidden vectors is zero if and only if it is zero when these components are replaced by indeterminates.  Moreover, for any particular choice of the clause vectors $\ket{v}$ we have $\rsat \ge \rgen$, since this choice can only result in linear dependences among the forbidden vectors and thus increase the rank of the satisfying subspace.

%This generic rank is an integer-valued function of $G$, and its computational complexity is an interesting open question.

%In particular, they point out that letting an adversary choose the signs of the literals in order to minimize the number of satisfying assignments gives an upper bound on the generic rank.  We call the number of satisfying assignments, minimized over all settings of the literals, the \emph{classical adversarial rank} $\radv$.

Our bounds work by partitioning random hypergraphs into gadgets for which we can compute $\rgen$ exactly.  In order to show that certain partitions exist, we use the technique of differential equations to analyze simple greedy algorithms.  To our knowledge, this is the first time that differential equations have been used to prove \emph{upper} bounds on satisfiability thresholds.

%We also consider the case where the clause vectors are chosen uniformly from the separable vectors, i.e., those of the form $\otimes_{i=1}^k \ket{v_i}$ where each $\ket{v_i}$ is a one-qubit state.  Our upper bounds apply to these separable formulas.  In this case we also prove a lower bound, by determining the exact threshold $\alpha_c^{\rm sep} \approx 0.917$ below which these formulas have separable satisfying assignments.

%This allows us to completely characterize the graphs for which the corresponding quantum 2-SAT problem is generically satisfiable.  For $k \ge 3$, we show how to prove upper bounds on $\alpha_c$ by For our lower bound, we determine the exact threshold for the existence of separable satisfying states in the case where the clause vectors are separable.

%\section{Separable clauses}
%\label{sec:separable}

%\begin{theorem}
%\label{thm:separable}
%Let $G$ be a $k$-uniform hypergraph.  In each clause choose $k$ vectors $\ket{v_i}$ uniformly from the norm-$1$ vectors in $\C_2$, and let the forbidden vector in that clause be $\ket{v} = \bigotimes_{i=1}^k \ket{v_i}$.  Then $\rsat=\rgen$ with probability $1$.
%\end{theorem}

%{\bf This is false :-( is it true for simple hypergraphs, or those where no clauses are repeated?}

\section{The case $k=2$}
\label{sec:k2}

As a warm-up, in this section we reproduce results of Laumann et al.~\cite{laumann} on quantum 2-SAT, determining $\rgen$ for all multigraphs, and in particular determining for which multigraphs the corresponding formula is generically satisfiable.

\begin{theorem}
\label{thm:k2}
Let $G$ be a connected multigraph with $n$ vertices and $m$ edges.
  If we form a quantum 2-SAT formula by replacing each edge $(i,j)$ with a clause forbidding a
  random vector $\ket{v_{ij}} \in \C_2^{\otimes 2}$, then its generic rank $\rgen$ is
\be
\label{R}
\rgen=\left\{ \ba{rcll} n+1 &\mbox{if}& m=n-1 & \mbox{\rm ($G$ is a tree)}, \\
2 &\mbox{if} & m=n & \mbox{\rm ($G$ is a cycle or a tree with a double edge)}, \\
1 &\mbox{if} & n=2 \quad \mbox{and} \quad m=3 & \mbox{\rm ($G$ consists of a triple edge)},\\
0 &\mbox{if} & n\ge 3 \quad \mbox{and} \quad m>n. & \\
\ea
\right.
\ee
\end{theorem}

\begin{proof}
Let  $V=\{1,\ldots,n \}$ be the vertices of $G$ and $E$ be its edges.
Let $\phi=\{ \ket{v_{ij}} \}_{(i,j)\in E}$ be a fixed instance of 2-SAT defined on $G$.
Here $\ket{v_{ij}} \in \C_2 \otimes \C_2$ is a forbidden state associated with
the edge $(i,j)$.

Let $O$ be an invertible local operator, or ILO---that is, $O=\bigotimes_{i\in V} O_i$ where all the $O_i$ are invertible, but not necessarily unitary.  Define a new instance
\[
O \cdot \phi = \{ O_i \otimes O_j \, \ket{v_{ij}} \}_{(i,j)\in E}.
\]
We claim that $\phi$ and $O \cdot \phi$ have the same rank. Indeed, a state $\ket{\psi}$
is a satisfying assignment for $\phi$ iff a state $(O^\dag)^{-1}\, \ket{\psi}$ is a satisfying assignment for $O\cdot \phi$.

If $\phi$ is a generic instance, all states $\ket{v_{ij}}$ are entangled. Let $T$ be any spanning tree of $G$.
Then there exists an ILO $O$ that maps all forbidden states on the edges of $T$ to singlets
$\frac{1}{\sqrt{2}} (\ket{0,1}-\ket{1,0})$~\cite{bravyi}.
The operator $O$ maps forbidden states on edges $(i,j) \notin T$ to some new forbidden states
which are still generic (although their new distribution might not be uniform). Thus it suffices to compute $\rgen$
for instances $\phi$ such that all edges of $T$ are singlets and all other edges are generic states.

Let $\phi_{\rm tree}$ be the restriction of $\phi$ onto the tree $T$ (all clauses $(i,j)\notin T$ are removed).
Any satisfying assignment of $\phi_{\rm tree}$
is invariant under transpositions of any pair of qubits $(i,j)\in T$, and thus invariant under
any permutation of qubits.  Obviously, the converse is also true.
Thus the satisfying subspace of $\phi_{\rm tree}$
is exactly the totally symmetric subspace $\calS_n \subset \C_2^{\otimes n}$, which has dimension $\dim \calS_n=n+1$.

Suppose $m=n-1$.  Then $G$ is a tree, $\phi=\phi_{\rm tree}$, and thus $\rgen=n+1$.

Suppose $m\ge n$.
Let $\ket{v_{ij}}$ be any clause of $\phi\backslash \phi_{\rm tree}$.
Since the satisfying assignments of $\phi$ span some subspace of $\calS_n$, the choice of $i$ and $j$ doesn't matter---applying the clause $\ket{v_{ij}}$ to any pair of qubits gives
%SBB: extra 'the'
the same rank.  Let us apply all clauses of $\phi\backslash \phi_{\rm tree}$ to the
pair of qubits $1, 2$.  There are $m-n+2$ forbidden states on this
pair of qubits: the singlet from $\phi_{\rm tree}$, and $m-(n-1)$ forbidden states from $\phi\backslash \phi_{\rm tree}$.

If $m\ge n+1$ then there are at least $3$ forbidden states on qubits $1,2$.
These completely fix a state of qubits $1,2$,
say $\ket{\omega_{1,2}}$.  In the generic case $\ket{\omega_{1,2}}$ is entangled.
If $n\ge 3$, the monogamy of entanglement implies that $\ket{\omega_{1,2}}$ cannot be symmetrically extended to $n$ qubits.  In that case, there are no satisfying assignments and $\rgen=0$.
If $n=2$ then $\ket{\omega_{1,2}}$ is the unique satisfying assignment, and $\rgen=1$.

It remains to consider the case $m=n$, where $G$ contains a single cycle or is a tree with a double edge.
Now qubits $1,2$ have two forbidden states: the singlet and some state $\ket{\psi_{1,2}}$.
We can get an upper bound on $\rgen$ by choosing  $\ket{\psi_{1,2}}$ adversarially, for example,
$\ket{\psi_{1,2}} = \ket{0,1}+\ket{1,0}$.  In this case we can look for satisfying assignments with a fixed number
of 1s, since all clauses commute with the particle number operator $\sum_{j=1}^n \outer{1}{1}_j$.
For any number of particles $0 \le m \le n$ there is only one symmetric state $\ket{S_m}$---the uniform superposition
of all binary strings with Hamming weight $m$.  One can easily check that $\ket{S_m}$ is orthogonal to
$\ket{0,1}+\ket{1,0}$ iff $m=0$ or $m=n$. Thus there are two satisfying assignments: $\ket{0^{\otimes n}}$
and $\ket{1^{\otimes n}}$.  This proves that $\rgen \le 2$.

To show that $\rgen \ge 2$, for any $\ket{\psi_{1,2}}$ we can try to construct a satisfying assignment $\ket{\varphi^{\otimes n}}$
for some $\ket{\varphi} \in \C_2$.  Without loss of generality $\ket{\psi_{1,2}}$ is symmetric, since otherwise it is some
linear combination of the singlet and a symmetric state and we can redefine the clause.  In the generic case $\ket{\psi_{1,2}}$ is also entangled.  Consider the following proposition:
\begin{prop}
For any symmetric entangled state $\ket{\psi} \in \C_2 \otimes \C_2$ there exist two linearly independent states
$\ket{\varphi}, \ket{\varphi'} \in \C_2$ such that
\be
\inner{\psi}{\varphi \otimes \varphi} = \inner{\psi}{\varphi'\otimes \varphi'} = 0 \, .
\ee
\end{prop}
\begin{proof}
Let $A_{ij}=\inner{\psi}{i,j}$ be the $2\times 2$ complex matrix corresponding to $\ket{\psi}$.  We are promised that $A$ is symmetric, $A^T=A$, and non-singular.  Using Gaussian elimination, for symmetric matrices one can find an invertible complex matrix $O$ such that
$OAO^T=\one$.  This is equivalent to
\be
(O \otimes O) \ket{\psi} = \ket{0,0} + \ket{1,1} \, .
\ee
% CM: added parentheses
Now we can choose
\be
\ket{\varphi}, \ket{\varphi'} = (O^\dag)^{-1} (\ket{0} \pm i \ket{1}) \, .
\ee
\end{proof}

\noindent
Thus $\rgen=2$, and the proof is complete.
\end{proof}

Now suppose that we form a random multigraph with $n$ vertices and $m=\alpha n$ edges by choosing uniformly with replacement from the ${n \choose 2}$ possible edges.  It is well known
%SBB: is there any reference for this?
that if $\alpha < 1/2$, then with high probability every connected component has at most one cycle, or one double edge, but never both---while if $\alpha > 1/2$, then with high probability there is a giant connected component with multiple cycles.  Thus, as already shown in~\cite{laumann}, random quantum 2-SAT has a phase transition from satisfiability to unsatisfiability at $\alpha = 1/2$.

On the other hand, the classical 2-SAT transition is at $\alpha = 1$ (see e.g.~\cite{2sat}).  This is a good illustration of the fact that the generic quantum problem is much more constrained.

\section{Expected gadget projectors}
\label{sec:gadget}

The next lemma generalizes a result of Laumann et al.~\cite{laumann}, which showed that addiing a $k$-clause reduces $\rgen$ by a factor of $1-2^{-k}$.  Our argument is somewhat simpler.

\begin{lemma}
\label{lem:gadget}
Let $G$ and $H$ be hypergraphs with $n$ and $t \le n$ vertices respectively.  Let $G \cup H$ denote the hypergraph resulting from adding a copy of $H$, on some subset of $G$'s vertices, to $G$.  Then
\[
\rgen(G \cup H) \le 2^{-t} \rgen(H) \rgen(G) \, .
\]
\end{lemma}

\begin{proof}
Let $\Pi_H$ be the projection operator onto the satisfying subspace of $H$, viewed as a subspace of $\C_2^{\otimes n}$.  Now consider its expectation $\Exp \Pi_H$, taken over the choice of clause vectors $\ket{v}$.  Since each $\ket{v}$ is chosen uniformly from the sphere in $\C_2^{\otimes k}$, and since the uniform measure is invariant under any rotation of a single qubit, $\Exp \Pi_H$
%SBB: any one-qubit unitary operator
commutes with any one-qubit unitary operator affecting a vertex in $H$.
Since $\Pi_H$ acts as the identity on the other $n-t$ vertices, it commutes with one-qubit unitary operators on them as well.

with any one-qubit operator affecting a vertex in $H$.  Since $\Pi_H$ acts as the identity on the other $n-t$ vertices, it commutes with one-qubit operators on them as well.

Thus $\Pi_H$ commutes with any Pauli operator.  Since these form a basis for the full matrix algebra acting on $\C_2^{\otimes n}$, it follows that $\Exp \Pi_H$ must be a scalar.  Since
\[
\rank \Pi_H = 2^{n-t} \rgen(H)
\]
holds with probability $1$ (where the factor of $2^{n-t}$ comes from being able to set the other qubits of $G$ arbitrarily), it also holds in expectation.  Thus
\[
\tr \Exp \Pi_H = \Exp \tr \Pi_H = 2^{n-t} \rgen(H) \, ,
\]
and therefore
\[
\Exp \Pi_H = 2^{-t} \rgen(H) \,\one \, .
\]

Now suppose the clause vectors of $G$ are in general position.  For any choice of clause vectors on $H$ we have
\[
\rgen(G \cup H)
\le \rank \Pi_{G \cup H}
\le \tr (\Pi_G \Pi_H \Pi_G) \, .
\]
(This follows from the fact that $ABA$ is positive whenever $A$ and $B$ are projection operators.)  This is also true in expectation over the clause vectors of $H$, so
\[
\rgen(G \cup H)
\le \Exp \tr (\Pi_G \Pi_H \Pi_G)
= 2^{-t} \rgen(H) \tr \Pi_G
= 2^{-t} \rgen(H) \rgen(G) \, ,
\]
completing the proof.
\end{proof}

As pointed out in~\cite{laumann}, if we take $H$ to be a single clause for which $\rgen = 2^k-1$, this shows that for a random formula with $n$ variables and $m=\alpha n$ clauses we have
\[
\rgen \le 2^n (1-2^{-k})^m = \left[ 2 (1-2^{-k})^\alpha \right]^n \, .
\]
If $\alpha > \log_{8/7} 2 \approx 5.191$ then this bound is exponentially small, showing that such formulas are unsatisfiable and placing an upper bound on the critical threshold.  A similar argument applies in the classical case.  However, in the next sections we will show that in the quantum case, we can prove much stronger upper bounds by computing $\rgen$ for larger gadgets.

\section{Two handy gadgets}

In this section we compute the generic rank exactly for two families of hypertrees.  We will use these calculations to derive our upper bounds on the critical threshold.

\subsection{The sunflower}
\label{sec:sunflower}

Consider the \emph{$(d,k)$-sunflower}, the $k$-uniform hypergraph consisting of $n$ clauses (edges), each pair sharing a common ``center'' vertex $z$.  Specifically, the graph is defined over the $1+d(k-1)$ vertices
%SBB: i and j are exchanged to make the notations consistent
\[
\{z\} \cup \{ x_i^j \mid 1 \leq j \leq d, 1 \leq i \leq k-1 \} \, ,
\]
and contains the $d$ clauses
\[
C_j = \{ z \} \cup \{ x_i^j \mid 1 \leq i \leq k-1 \}\,.
\]
See Fig.~\ref{fig:sunflower} for an example.

\begin{figure}
\begin{center}
\includegraphics[width=0.4 \textwidth]{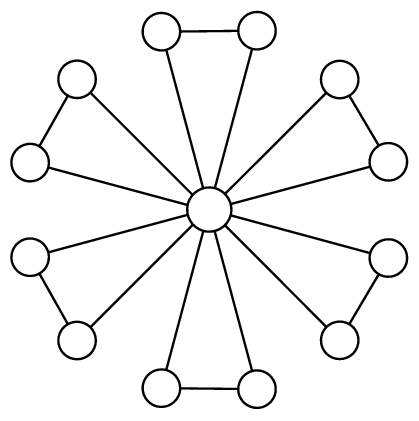}
\end{center}
\caption{The $(6,3)$-sunflower.}
\label{fig:sunflower}
\end{figure}

%For any separable family of clause vectors, we may introduce a basis $\{\ket{0}, \ket{1}\}$ for the copies of $\C^2$ associated with the vertices $x_i^j$ so that the clause vectors take the form
%\[
%\ket{c_i} = \ket{z_i} \otimes \ket{0}^{\otimes k-1}\,.
%\]

\begin{lemma}
\label{lem:sunflower}
Let $S(d,k)$ denote the generic rank of the $(d,k)$-sunflower. Then
$$
S(d,k) = 2(2^{k-1} - 1)^d  \left( \frac{d}{2^k-2} + 1 \right) \, .
%= d(2^{k-1} - 1)^{d-1} + 2(2^{k-1} - 1)^d \,.
$$
\end{lemma}

\begin{proof}
Decompose the Hilbert space of  the $(d,k)$-sunflower as
\[
\calH=\calH_0 \otimes \calH_1 \otimes \cdots \otimes \calH_d \, ,
\]
where
\[
\calH_0 = \C^2, \quad \calH_1= \cdots =\calH_d=(\C^2)^{\otimes k-1} \, .
\]
Here $\calH_0$ describes the central qubit and $\calH_j$ describes the other $k-1$ qubits on the $j$th petal for
 %SBB: a typo
$1 \le j \le d$.  Let $C_j$ denote the clause on the $j$th petal, and let $\ket{v_j} \in \calH_0 \otimes \calH_j$ be its forbidden state.
Clearly $\ket{v_j}$ has at most two non-zero Schmidt coefficients, so we can always choose a unitary
operator $U_j$ acting on $\calH_j$ such that
\be
\ket{v_j}=(I_0\otimes U_j )\, \ket{u_j} \otimes \ket{0^{\otimes k-2}} \, .
\ee
Here $\ket{u_j}$ is some entangled state between the central qubit and the first qubit of the $j$th petal.
Let us refer to the remaining  $k-2$ qubits of the $j$th petal that are projected onto $\ket{0}$ as
ancillas. Clearly $U_j$ does not change the rank, so without loss of generality we can assume that $U_j=\one$ for all $j$.

If we ignore the ancillas, then we get an instance of quantum 2-SAT on a star graph with $d$ edges.  By Theorem~\ref{thm:k2}, its
generic rank is $d+2$. Clearly ancillas of the $j$th petal can be ignored iff they are all set to the state $\ket{0}$.
Let us say that such a petal is {\em active}.  Otherwise, if at least one ancilla of the $j$th petal is $\ket{1}$, it becomes {\em inactive} because
the corresponding clause is already satisfied.  An inactive petal contributes a factor of $2^{k-1}-2$ to the rank.
For a fixed subset $A$ of active petals, the generic rank is
\be
\label{RA}
R(A)=(a+2)(2^{k-1}-2)^{d-a}, \quad a \equiv |A| \, .
\ee
Therefore, the generic rank for the $(d,k)$-sunflower is
\be
\label{RA1}
R=\sum_{A\subseteq \{1,\ldots,d\}}  R(A) = \sum_{a=0}^d {d \choose a} (a+2)(2^{k-1}-2)^{d-a}
%= (2^{k-1}-1)^{d-1}(d+2(2^{k-1}-1)) \, ,
= 2(2^{k-1} - 1)^d  \left( \frac{d}{2^k-2} + 1 \right) \, .
\ee
This completes the proof.
\end{proof}

%Upper bound. Consider the expansion of a satisfying vector $\ket{s}$ in the product basis above:
%$$
%\ket{s} = \sum_{\vec{b}} \alpha_{\vec{b}} \ket{\vec{b}} \otimes \ket{z_{\vec{b}}}\,,
%$$
%this sum taken over the $2^{d(k-1)}$ elements $\ket{b_1^1} \otimes \cdots \otimes \ket{b_{k-1}^d} = \vec{b} \in \{0,1\}^{d(k-1)}$ and the tensor product arranged so that the coordinate associated with the center vertex appears last. We say that $\vec{b}$ \emph{satisfies} the clause $C_i$ if $\bigvee_j b_i^j = 1$, in which case we write $\vec{b} \vdash C_i$. Observe, now, that if $\vec{b}$ does not satisfy $C_i$, the vector $\ket{z_\vec{b}}$ must be perpendicular to $\ket{z_i}$; otherwise $\Pi_{C_i} \ket{s}$ is surely nonzero, where $\Pi_{C_i}$ is the clause projection operator. If the $\ket{z_i}$ are in general position, a given $\ket{z_{\vec{b}}}$ can be orthogonal to no more than one of them and we may write
%$$
%\ket{s} = \sum_{\forall i, \vec{b} \vdash C_i} \alpha_{\vec{b}} \ket{\vec{b}} \otimes \ket{z_{\vec{b}}} + \sum_i \sum_{\substack {\vec{b} \not\vdash C_i \\\forall j \neq i, \vec{b} \vdash C_i}} \alpha_{\vec{b}} \ket{\vec{b}} \otimes \ket{z_i}^\perp\,.
%$$
%Note, also, that any vector of this form lies in the satisfying space, which evidently has dimension precisely
%$$
%S(d,k)
%= d(2^{k-1} - 1)^{d-1} + 2(2^{k-1} - 1)^d
%= 2(2^{k-1} - 1)^d  \left( \frac{d}{2^k-2} + 1 \right)
%\, .
%$$

\subsection{The nosegay}
\label{sec:nosegay}

For another example, consider a \emph{nosegay}, as shown in Fig.~\ref{fig:nosegay}.  It consists of a single edge, where each of its vertices has some number of additional edges attached to it.

\paragraph{$3$-uniform nosegays} Let us momentarily restrict our attention to the case $k=3$.  Then an $(a,b,c)$-nosegay has $a$, $b$, and $c$ additional edges, for a total of $a+b+c+1$ edges and $3+2(a+b+c)$ vertices.

\begin{figure}
\begin{center}
\includegraphics[width=0.4 \textwidth]{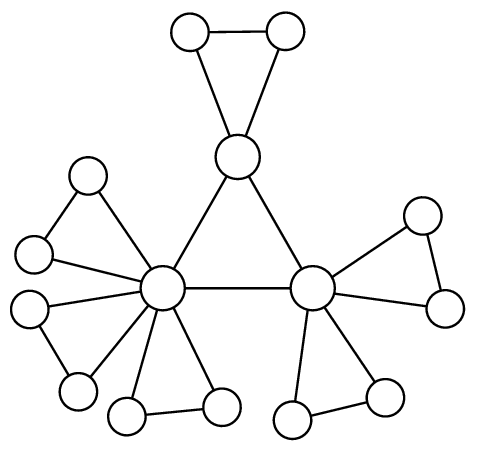}
\end{center}
\caption{The $(1,2,3)$-nosegay.}
\label{fig:nosegay}
\end{figure}

\begin{lemma}
\label{lem:nosegay}
Let $R_{(a,b,c)}$ denote the generic rank of the 3-uniform $(a,b,c)$-nosegay.  Then
\[
R_{(a,b,c)}
= 3^{a+b+c-3} \left[ (a+6)(b+6)(c+6)-(a+3)(b+3)(c+3) \right]
%= 3^{a+b+c-2} \left( ab+bc+ca + 9(a+b+c) + 63 \right) \, .
\]
\end{lemma}

\begin{proof}
Let us label the qubits of the central triangle by $1$, $2$ and $3$. These qubits have $a,b$ and $c$ hanging triangles
attached to them respectively. Each triangle represents a generic forbidden $3$-qubit state.

Let us also define the $[a,b,c]$-nosegay: it coincides with the $(a,b,c)$-nosegay except that
each hanging triangle is replaced by a hanging edge. Accordingly, the $[a,b,c]$-nosegay has
$n=3+a+b+c$ qubits. Each hanging edge represents a generic forbidden $2$-qubit state while the central
triangle represents a generic forbidden $3$-qubit state.

Let $R_{(a,b,c)}$ and $R_{[a,b,c]}$ be the generic ranks of the $(a,b,c)$-nosegay and the $[a,b,c]$-nosegay.
Repeating the arguments used to compute the genetic rank of sunflowers in Lemma~\ref{lem:sunflower}, we get
\be
\label{simplified}
R_{(a,b,c)}=\sum_{p=0}^a \sum_{q=0}^b \sum_{r=0}^c 2^{a+b+c-p-q-r}\, { a \choose p} \,
{ b \choose q}\, { c \choose r}  R_{[p,q,r]}.
\ee
In the rest of the section we prove that
\be
\label{Rabc}
R_{[a,b,c]}=(a+2)(b+2)(c+2)-(a+1)(b+1)(c+1)
\ee
which after simple algebra yields
\be
R_{(a,b,c)}=3^{a+b+c-3}\left[
(a+6)(b+6)(c+6)-(a+3)(b+3)(c+3)\right].
\ee
We shall start from using the symmetry of the $[a,b,c]$-nosegay to bring the forbidden states into a canonical form
such that the forbidden state associated with the central triangle is
\be
\label{can1}
\ket{v_{1,2,3}}=\frac{1}{\sqrt{2}}\, ( \ket{0,0,0} - \ket{1,1,1})
\ee
while the forbidden state associated with any hanging edge $(i,j)$ is the singlet,
\be
\label{can2}
\ket{v_{ij}}=\frac1{\sqrt{2}}\, ( \ket{0,1} - \ket{1,0}).
\ee
We claim that any set of generic forbidden states can be mapped to the ones defined in~\eqref{can1}, \eqref{can2} by applying invertible
local operators (ILO) to every qubit. Indeed, it was shown by D\"ur, Vidal, and Cirac~\cite{DVC:2000}
that a generic $3$-qubit state is ILO-equivalent to the GHZ state which is in turn ILO-equivalent to
the state  $\ket{v_{1,2,3}}$ defined in~\eqref{can1}.
Note that applying an ILO at this step maps the forbidden states on the hanging edges to some entangled
$2$-qubit states. We can convert them to singlets by applying proper ILO to the free end of every hanging edge.
As we argued in Theorem~\ref{thm:k2}, the dimension of the satisfying subspace is invariant under ILO, so it suffices
to compute $R_{[a,b,c]}$ with forbidden states defined by~\eqref{can1} and \eqref{can2}.

Our arguments will rely on the fact that the canonical forbidden states define an instance of {\em stoquastic}
3-SAT studied in~\cite{BT:2008}. An instance of stoquastic 3-SAT
is defined by a family of projectors $\phi=(\Pi_1,\ldots,\Pi_m)$ which have real non-negative matrix elements in the computational basis, such that every projector acts on at most $3$ qubits. A state $\ket{\psi}$ is a satisfying assignment for $\phi$
iff $\Pi_a\, \ket{\psi} =\ket{\psi}$ for every $a=1,\ldots,m$. In our case the instance $\phi$ is defined by
a family of projectors
%SBB: a typo
\be
\Pi_{1,2,3}= \one-\outer{v_{1,2,3}}{v_{1,2,3}}, \quad \Pi_{i,j}= \one-\outer{v_{ij}}{v_{ij}}
\ee
where $\Pi_{1,2,3}$ acts on the central triangle and $\Pi_{i,j}$ acts on every hanging edge $(i,j)$.
A direct inspection shows that matrix elements of the above projectors have the following properties:
%SBB: rephrased
\begin{itemize}
\item Any off-diagonal matrix element belongs to the set $\{0,1/2\}$
\item Any diagonal matrix element belongs to the set $\{1,1/2\}$
\end{itemize}
A family of projectors with such properties defines an instance $\phi$ of {\em simplified stoquastic 3-SAT}~\cite[Section~6.3]{BT:2008}. In particular, it was shown in~\cite{BT:2008} that the number of satisfying assignments of $\phi$ is equal to the number of connected components of a graph
$G=(V,E)$, where $V=\{0,1\}^n$ and $(x,y)\in E$ iff there exists a projector $\Pi_a\in \phi$ such that
$\innerm{x}{\Pi_a}{y}=1/2$. (A satisfying assignment associated with a connected component $V_\alpha \subseteq V$
is the uniform superposition of all vertices in $V_\alpha$, see~\cite{BT:2008} for details.)
It remains to count connected components in the graph $G$ associated with the $[a,b,c]$-nosegay.

Let us partition the qubits of the nosegay into $3$ disjoint subsets $A,B,C$ such that
qubit $1$ and the free ends of all hanging edges attached to it form $A$,
qubit $2$ and the free ends of all hanging edges attached to it form $B$,
and the remaining qubits form $C$. By definition,
\be
|A|=n_a=a+1, \quad |B|=n_b=b+1, \quad |C|=n_c=c+1.
\ee
Using the projectors that live on the hanging edges we conclude  that
$(x,y)\in E$ whenever the restrictions of $x$ and $y$ onto any of the subsets $A$, $B$, $C$ have the
same Hamming weight. Leaving out the projector $\Pi_{1,2,3}$ temporarily we can thus label the connected
%SBB:
components of $G$ by triples of integers
\be
(\alpha,\beta,\gamma), \quad 0\le \alpha \le n_a, \quad
 0\le \beta \le n_b, \quad
  0\le \gamma \le n_c,
\ee
%SBB: more explanation is added
where for any vertex $x\in V$ we define $\alpha,\beta$ and $\gamma$ as the Hamming weight of $x|_A$, $x|_B$, and $x|_C$  respectively.
Adding the projector $\Pi_{1,2,3}$ adds extra edges to the graph $G$, which glue together
some components according to the following rules:
\be
(\alpha,\beta,\gamma)\sim (\alpha+1,\beta+1,\gamma+1), \quad (\alpha,\beta,\gamma)\sim (\alpha-1,\beta-1,\gamma-1).
\ee
Thus the number of connected components of $G$ is the same as the number of diagonals
parallel to the axis $(1,1,1)$ in a cube of size $[0,n_a]\times[0,n_b]\times [0,n_c]$.  This yields~\eqref{Rabc} and completes the proof.
\end{proof}

\paragraph{$k$-uniform nosegays} A $k$-uniform nosegay is determined by a vector $\vec{d} = (d_1, \ldots, d_k)$ of nonnegative integers.  It is the $k$-uniform hypergraph given by a single edge $A = \{ \alpha_1, \ldots, \alpha_k\}$, the $i$th vertex of which is incident upon $d_i$ other ``hanging'' edges $B_i^1, \ldots, B_i^{d_i}$ with $B_i^j = \{ \alpha_i \} \cup \{ \beta_i^j(\ell) \mid 1 \leq \ell \leq k-1\}$. These ``hanging'' edges intersect the central edge at a unique vertex, and are otherwise all disjoint. We refer to such a graph as a $\vec{d}$-nosegay.  In the next lemma we determine the rank of its satisfying space when adorned with separable clause vectors, which is an upper bound on its rank in the generic case.

\begin{lemma}
Let $N(\vec{d})$ denote the rank of the satisfying subspace of the $\vec{d}$-nosegay when adorned with separable clause vectors in general position. Then
$$
N(\vec{d}) = \prod_i (2^{k-1} - 1)^{d_i-1} \left[ \prod_i \left(d_i + 2\left({2^{k-1}} -1\right)\right) - \prod_i \left(d_i + (2^{k-1}-1)\right) \right]\,.
$$
\end{lemma}
\begin{proof}
As the clause vectors are separable, we may introduce a basis for the Hilbert spaces (copies of $\C^2$) associated with the vertices $\beta_i^j(\ell)$ so that the clause vector associated with $B_i^j$ has the form $\ket{a_i^j} \otimes \ket{0}^{k-1}$. (Here the first factor in this tensor product is associated with the vertex $\alpha_i$.) As with the sunflower, we may expand a satisfying vector $\ket{v}$ according to this basis:
$$
\ket{v} = \sum_{\vec{b}} \ket{\vec{b}} \otimes \ket{v_\vec{b}}\,,
$$
where $\ket{\vec{b}} = \ket{b_1} \otimes \ket{b_2} \otimes \cdots$ is a basis vector in the tensor product of the Hilbert spaces associated with the vertices $b_i^j(\ell)$ and $v_\vec{b}$ lies in the Hilbert space $\mathcal{H}_A$ associated with the center edge.

We write $\vec{b} \vdash B_i^j$ when $b_t = 1$ for one of the indices associated with $B_i^j$. Should $\vec{b} \not\vdash B_i^j$, observe that we may expand $v_\vec{b}$ in a Schmidt decomposition $\sum_{s=1}^2 \ket{v_s} \otimes \ket{w_s}$, where the $\ket{v_s}$ lie in the Hilbert space associated with $\alpha_i$, and, considering that the $\ket{w_s}$ are orthogonal and that $\ket{v}$ satisfies the clause $B_i^j$, conclude that each $\langle {v_s}, a_i^j \rangle = 0$. It follows that $\ket{v_\vec{b}}$ has the form $\ket{a_i^j}^\perp \otimes \ket{w}$ (where $\ket{a_i^j}^\perp$ is orthogonal to $\ket{a_i^j}$). As the $\ket{a_i^j}$ are in general position, then, for each $i$ we must have $|\{ B_i^j \mid \vec{b} \vdash B_i^j\}| \leq 1$. Additionally, observe that if, for each $i$, we have $\vec{b} \vdash B_i^j$ for some $j$ then $v_\vec{b}$ is completely determined (and, in fact, separable), and cannot satisfy the clause $A$. In general, writing $k - \ell = |\{ B_i^j \mid \vec{b} \vdash B_i^j \}|$, we find that orthogonality with the vector associated with $A$ precisely constrains $v_\vec{b}$ to a subspace of $\mathcal{H}_A$ of dimension $2^\ell - 1$. Evidently, this expresses the satisfying subspace as an orthogonal direct sum of subspaces of total dimension
$$
\sum_{I \subset \{ 1, \ldots, k\}} \left( \prod_{i \in I} d_i (2^{k-1} - 1)^{d_i-1} \right) \cdot \left(\prod_{i \not\in I} \left(2^{k-1} - 1\right)^{d_i}\right) \cdot \left(2^{k - |I|} -1\right)\,,
$$
equal to the expression in the statement of the lemma.
\end{proof}

\section{Upper bound on the critical density}
\label{sec:critical-upper}

In this section we present two upper bounds on the critical threshold.  The first one is weaker but simpler.

\begin{theorem}
\label{thm:critical-upper1}
Let $H$ be a random $3$-uniform hypergraph with $n$ vertices and $m=\alpha n$ edges.  If $\alpha > 3.894$, then with high probability the corresponding quantum 3-SAT problem is unsatisfiable.
\end{theorem}

\begin{proof}
Let $H$ be a $k$-uniform hypergraph with $n$ vertices and $m$ edges, corresponding to a formula with $m$ clauses.  If we can partition $H$ into a set of sunflowers, where there are $n_d$ sunflowers of each degree $d$, then by Lemmas~\ref{lem:gadget} and~\ref{lem:sunflower}, the generic rank of its satisfying subspace is bounded by
\begin{align}
\rgen
&\le 2^n \prod_{d=1}^\infty \left( \frac{S(d,k)}{2^{1+d(k-1)}} \right)^{n_d} \nonumber \\
&= 2^n \prod_{d=1}^\infty \left( \left(1-\frac{1}{2^{k-1}}\right)^d \left( \frac{d}{2^k-2} + 1 \right) \right)^{n_d}
\label{eq:sunflower-bound} \\
%&= 2^n \left(1-\frac{1}{2^{k-1}}\right)^m \prod_{d=1}^\infty \left( \frac{d}{2^k-2} + 1\right)^{n_d} \\
&= 2^n  \prod_{d=1}^\infty \left( \left( \frac{3}{4} \right)^d \left( \frac{d}{6} + 1\right) \right)^{n_d} \, .
\nonumber
\end{align}
where in the last line we set $k=3$.
%we use the fact that $\sum_d d n_d = m$ and then

Clearly~\eqref{eq:sunflower-bound} is minimized if we have a small number of sunflowers of high degree.  Ideally, we would like to characterize the best possible such partition.  For now, we content ourselves with the partition resulting from the following simple algorithm: at each step, choose a random vertex, declare it and its edges to be a sunflower, and remove them from the graph.

We can carry out this algorithm in continuous time, by assigning each vertex a uniformly random index $t \in [0,1]$ and removing
%SBB: style change
vertices in the order of decreasing $t$.  In that case, by the time we remove a vertex $v$ with degree $t$, its sunflower includes those edges whose other vertices all have index less than $t$.  Since this is true of each of its clauses independently with probability $t^{k-1}$, and since the original degree
%SBB: ``degree distribution of" should be  ``degree distribution of $v$"
 distribution of $v$  is Poisson with mean $k\alpha$, the degree of $v$'s sunflower
%SBB: more explanation is added
at the moment when it is removed is Poisson-distributed with mean $k \alpha t^{k-1}$.  Integrating over $t$, the expected number of vertices whose sunflowers have degree $d$ is
 %SBB: the number of vertices is $n a_d$  rather than $a_d$ ?
 $na_d$, where
\begin{align}
a_d
&= \int_0^1 \frac{\e^{-k \alpha t^{k-1}} (k \alpha t^{k-1})^d}{d!} \,\dt \nonumber \\
&= \int_0^1 \frac{\e^{-3 \alpha t^2} (3 \alpha t^2)^d}{d!} \,\dt \nonumber \\
&= \frac{\Gamma(d+1/2)-\Gamma(d+1/2,3\alpha)}{2 \sqrt{3\alpha} \,d!} \, ,
\label{eq:sunflower-int}
\end{align}
where $\Gamma(a,z) = \int_z^\infty x^{a-1} \e^{-x} \,\dx$ is the incomplete Gamma function.

We then upper bound $\rgen$ by cutting off the product above $\dmax = 100$, ignoring the effect of the tiny fraction of sunflowers of greater degree.  Standard Azuma-type inequalities tell us that, with high probability, the number of sunflowers of degree $d$ is $a_d n + o(n)$ for all $d \le \dmax$.  Also, with high probability there are less than $\log n$ pairs of edges which share more than one vertex.  Their neighborhoods consist of sunflowers with two petals stuck together.  We claim that such a sunflower has lower rank than a normal one, but in any case pretending that these petals are not stuck together only changes the rank by a constant factor, and the effect of $\log n$ such steps changes the rank by $\poly(n)$.  So, with high probability,
\[
\rgen \le \poly(n) \left[ 2 \prod_{d=0}^\dmax \left( \left( \frac{3}{4} \right)^d \left( \frac{d}{6} + 1 \right) \right)^{a_d} \right]^n \, ,
\]
and therefore
\[
\lim_{n \to \infty} \frac{1}{n} \ln \rgen
\le \ln 2 + \sum_{d=0}^\dmax a_d \left( d \ln \frac{3}{4} + \ln \left(\frac{d}{6} + 1\right) \right) \, .
\]
If we set $\alpha = 3.894$, we find that this limit is $-1.372 \times 10^{-4}$, so $\rgen$ is exponentially small.  Thus with high probability in $H$, $\rgen = 0$ with probability $1$ in the clause vectors, and the formula is unsatisfiable.
\end{proof}

Next, we improve this result by partitioning the graph into nosegays instead of sunflowers.  Although the analysis is slightly harder, the algorithm is equally simple.

\begin{theorem}
\label{thm:critical-upper3}
Let $H$ be a random $3$-uniform hypergraph with $n$ vertices and $m=\alpha n$ edges.  If $\alpha > 3.594$, then with high probability the corresponding quantum 3-SAT problem is unsatisfiable.
\end{theorem}

\begin{proof}
At each step we choose a uniformly random edge, declare it and the edges it shares a vertex with to be a nosegay, and remove them and its vertices from the hypergraph.  The remaining hypergraph has $3$ fewer vertices.  Moreover, if we condition on how many edges it has, it is uniformly random in the model where edges are chosen with replacement.  This allows us to model this process with differential equations~\cite{wormald}.

If $t$ is the number of steps we have taken so far, then there are $n-3t$ remaining vertices.  Let $m$ denote the number of remaining edges.  Its expected change on each step is
\[
\Exp\!\left[ \Delta m \right] = -1 - \frac{9m}{n-3t} \, .
\]
Now we write $m=\mu n$ and $t=\tau n$, and rescale this to give a differential equation:
\begin{equation}
%\label{nosegay-diffeq}
\frac{\dmu}{\dtau} = -1-\frac{9 \mu}{1-3\tau} \, .
\end{equation}
Changing variables to the fraction $\nu = 1-3\tau$ of vertices remaining, this is
\begin{equation}
\label{nosegay-diffeq}
\frac{\dmu}{\dnu} = \frac{1}{3} + \frac{3 \mu}{\nu} \, .
\end{equation}
With the initial condition $\mu(1) = \alpha$, the solution to this is
\begin{equation}
\label{eq:mu}
\mu(\nu) = \frac{\nu}{6} \left( (6 \alpha + 1) \nu^2 - 1 \right) \, .
\end{equation}
This becomes zero when $\nu_0$ vertices are left, where
\[
\nu_0 = \frac{1}{\sqrt{6 \alpha+1}} \, .
\]
At that point, there are no edges left, and the algorithm stops.

With high probability, for any $\nu > \nu_0$, the number of edges remaining when there are $\nu n$ vertices left is $m(\nu) = \mu(\nu) n + o(n)$.  Similar to the proof of Theorem~\ref{thm:critical-upper1}, summing over the $(1-\nu_0)n/3 + o(n)$ steps of the algorithm then gives
\[
\lim_{n \to \infty} \frac{1}{n} \ln \rgen
\le \ln 2 + \frac{1}{3} \int_{\nu_0}^1 \Exp_{a,b,c}\!\left[ \ln \frac{R_{(a,b,c)}}{2^{3+2(a+b+c)}} \right] \dnu
\]
where $R_{(a,b,c)}$ is given by Lemma~\ref{lem:nosegay}, and where $a$, $b$, and $c$ are chosen according to independent Poisson distributions with mean $3 \mu / \nu$.

If we set $\alpha = 3.594$ and upper bound the expectation over $a$, $b$, and $c$ by ignoring terms where any of them is greater than $50$, then evaluating the resulting integral numerically we find that limit is $-1.601 \times 10^{-4}$.  Again $\rgen$ is exponentially small, so these formulas are unsatisfiable with high probability.
\end{proof}

There are a number of potential ways to improve this result.  First, we can achieve a better partition of the graph into gadgets by prioritizing high-degree vertices.  Analogous to~\cite{ach-moore}, we can analyze the resulting partition using a system of coupled differential equations, using the configuration model to keep track of the random graph conditioned on its degree distribution.  For the sunflower, this gives a bound of $3.689$---a significant improvement over Theorem~\ref{thm:critical-upper1}, but not as good as Theorem~\ref{thm:critical-upper3}.  We have not attempted a partition into nosegays that prioritizes high-degree clauses.

Second,we have no obligation to consider partitions into sunflowers or nosegays that can be found in polynomial time.  We could also use non-algorithmic proofs that a desirable partition exists.  These algorithms simply happen to be both efficient and easy to analyze.

Thirdly, we could use notions of local maximality which have been successful in the classical case, but it is not obvious how to apply these in the quantum setting.  When is an entangled satisfying state locally maximal?

\section{A upper bound for general $k$}
\label{sec:general-k}

In this section we use our sunflowers to prove an upper bound on the quantum $k$-SAT threshold for general $k$.  We have made no attempt to optimize this bound beyond the simplest possible argument, but it establishes that the quantum threshold is strictly less than the classical one for all $k \ge 6$.

\begin{theorem}
\label{thm:general-k}
Let $b \approx 0.573$ be the unique positive root of the equation $\ln 2 - 2b + \ln(b+1) = 0$.  Then for all $k \ge 3$, if $\alpha \ge 2^k b$ then with high probability the corresponding quantum $k$-SAT problem is unsatisfiable.
\end{theorem}

\begin{proof}
First we rewrite~\eqref{eq:sunflower-bound} as follows:
\[
\rgen
\le 2^n \left(1-\frac{1}{2^{k-1}}\right)^{\!m} \prod_{d=1}^\infty  \left( \frac{d}{2^k-2} + 1 \right)^{n_d} \, ,
\]
since $\sum_d n_d d = m$.  Treating the product as a harmonic mean over the vertices, bounding it as an arithmetic mean, and using the fact that the mean degree of a sunflower is $\Exp[d]=\alpha$ then gives
\begin{align}
\frac{1}{n} \ln \rgen
&\le \ln 2 + \alpha \ln \left(1-\frac{1}{2^{k-1}}\right) + \Exp\!\left[ \ln \left( \frac{d}{2^k-2} + 1 \right) \right]
\nonumber \\
&\le \ln 2 + \alpha \ln \left(1-\frac{1}{2^{k-1}}\right) + \ln \left( \frac{\alpha}{2^k-2} + 1 \right) \, .
\label{eq:general-k-exact}
\end{align}
Rearranging, applying the Taylor series of $\ln (1-x)$, and setting $\alpha = 2^k b$, we have
\begin{align*}
\frac{1}{n} \ln \rgen
&\le \ln 2 + (\alpha-1) \ln \left( 1-2^{1-k} \right) + \ln \left( 2^{-k} \alpha + 1 - 2^{1-k} \right) \\
&< \ln 2 - (\alpha-1) (2^{1-k} + 2^{1-2k}) + \ln \left( 2^{-k} \alpha + 1\right) - \frac{2^{1-k}}{2^{-k} \alpha + 1} \\
&= \ln 2 - 2 b + \ln \left( b + 1\right) + 2^{1-k} \left( 1 - b - \frac{1}{b + 1} \right) \\
&= 2^{1-k} \left( 1 - b - \frac{1}{b + 1} + 2^{-k} \right) \\
&< 0 \mbox{ for all $k \ge 3$} \, ,
\end{align*}
since $1-b-1/(b+1) < -1/8$.
\end{proof}

In contrast, the classical $k$-SAT threshold is known to be $(1+o(1)) 2^k \ln 2$~\cite{ach-peres}.  Since $b < \ln 2$, it follows that the quantum threshold is less than the classical one for sufficiently large $k$.  In fact, explicit lower bounds on the classical threshold for finite $k$ from~\cite{ach-peres} are greater than the upper bounds on the quantum threshold obtained by setting~\eqref{eq:general-k-exact} to zero for $k \ge 6$.  Using the approach of Theorem~\ref{thm:critical-upper1} improves this to $k \ge 5$.

%This leaves us with an interesting open question.  Assuming that a satisfiability threshold exists, is $\alpha_c^q$ proportional to $2^k$?  If so, what is $b = \lim_k \alpha_c^q / 2^k$?  The value of $b$ given in Theorem~\ref{thm:general-k} is almost certainly an overestimate.  Note that our current lower bounds tend to $1$ as $k \to \infty$, so we do not even know that $\alpha_c^q$ grows without bound.

%\section{Separable clauses}
%\label{sec:separable}

%SBB: I think we should include this section. I made some changes below.
\section{Open questions}
\label{sec:open}

We close with several open questions.
\begin{itemize}

\item What is the computational complexity of determining the generic rank of a hypergraph?  It would be surprising if it were not at least \numP-hard, but it is not obvious that it is in $\numP$.  On the other hand, we are not aware of any proof that it is even \NP-hard.

%SBB:
\item Can we prove a lower bound on the satisfiability threshold which is greater than the density at which a random graph contains a non-vanishing 2-core?  In particular, is there a phase where random formulas are satisfiable, but all satisfying states are entangled?
% CM: combined these two

\item Assuming that the quantum $k$-SAT threshold exists, is $\alpha_c^q$ proportional to $2^k$?  If so, what is $b = \lim_k \alpha_c^q / 2^k$?  The value of $b$ given in Theorem~\ref{thm:general-k} is almost certainly an overestimate.  Note that our lower bounds based on the existence of the 2-core as given in~\cite{laumann} actually decrease as $k$ increases---for instance, for $k=4$, $5$, and $6$ we have the lower bounds 0.772, 0.701, and 0.637.  So, at present, we do not even know that $\alpha_c^q$ grows without bound.
% CM: decrease
\end{itemize}

\section*{Acknowledgments}

S.B. received
support from the DARPA QUEST program under contract
no. HR0011-09-C-0047 and  is grateful to CWI for hospitality while this work
was being done.  
C.M. and A.R. are supported by the NSF under grant CCF-0829931, and by the DTO under contract W911NF-04-R-0009.

%%%%%%%%%%%%%%%%%%%%%%%


\begin{thebibliography}{99}

\bibitem{bravyi} S. Bravyi, Efficient algorithm for a quantum analogue of {2-SAT}. Preprint, {\tt quant-ph/0602108}.

\bibitem{valiant} L.G. Valiant, The complexity of enumeration and reliability problems.  \emph{SIAM J. Comput.} {\bf 8}(3) 410--421 (1979).

\bibitem{3satlower1} Alexis Kaporis, Lefteris Kirousis, Efthimios Lalas, Selecting complementary pairs of literals. Proc. LICS Workshop on Typical Case Complexity and Phase Transitions, 2003.

\bibitem{3satlower2} Mohammad Hajiaghayi and Gregory Sorkin, The satisfiability threshold for random 3-SAT is at least $3.52$.  Preprint, {\tt citeseer.ist.psu.edu/hajiaghayi03satisfiability.html} (2003).

%\bibitem{3satupper}  Olivier Dubois, Yacine Boufkhad, and Jacques Mandler, Typical random 3-{SAT} formulae and the satisfiability threshold.  Proc. 11th ACM-SIAM Symposium on Discrete Algorithms, 126--127 (2000).

\bibitem{3satupper} J. D\'iaz, L. Kirousis, D. Mitsche, and X. P\'erez-Gim\'enez, A new upper bound for 3-SAT.  Proc. FSTTCS 2008, 163--174.

\bibitem{mezard} M. Mezard, G. Parisi, and R. Zecchina, \emph{Science} {\bf 297}, 812 (2002).

\bibitem{laumann} C.R. Laumann, R. Moessner, A. Scardicchio, and S.L. Sondhi.  Phase transitions and random quantum satisfiability.  Preprint, {\tt arXiv:0903.1904}.

\bibitem{ach-peres} D. Achlioptas and Y. Peres, The threshold for random $k$-SAT is $2^k \ln 2 - O(k)$. Proc. STOC 2003, 223--231.

\bibitem{molloy} M. Molloy, Cores in random hypergraphs and Boolean formulas.  \emph{Random Struct. Algorithms} {\bf 27}(1) 124--135 (2005).

\bibitem{2sat} V. Chv\'atal and B. Reed. Mick gets some (the odds are on his side), Proc. 33rd Symposium
on the Foundations of Computer Science, 620--627 (1992).

\bibitem{DVC:2000}
W. D\"ur, G. Vidal, and J. I. Cirac,
Three qubits can be entangled in two inequivalent ways.
Phys.~Rev.~{\bf A 62}, 062314 (2000).

\bibitem{BT:2008}
S. Bravyi and  B. Terhal,
Complexity of stoquastic frustration-free Hamiltonians.
Preprint, {\tt arXiv:0806.1746}.

\bibitem{wormald}  Nicholas C. Wormald,
Differential equations for random processes and random graphs.
{\em Annals of Applied Probability} {\bf 5}, 1217--1235 (1995).

\bibitem{ach-moore} D. Achlioptas and C. Moore, ``Almost all graphs of degree 4 are 3-colorable.'' {\em Journal of Computer and System Sciences}, {\bf 67} (2003) 441--471.  Invited paper in special issue for STOC 2002.


\end{thebibliography}
\end{document}